\pdfoutput=1
\documentclass[%
 reprint,
 amsmath,amssymb,
 aps,
]{revtex4-1}

\usepackage{graphicx,amsmath,amssymb,amsfonts,braket,amsthm}
\usepackage{dcolumn}
\usepackage{bm}
\usepackage[bottom]{footmisc}

\theoremstyle{definition}
\newtheorem{lemma}{Lemma}

\newtheorem{observation}{Observation}

\begin{document}


\title{Geometric criterion for separability based on local measurement}

\author{Aryaman A. Patel}
\email{13me121.aaryaman@nitk.edu.in}
\author{Prasanta K. Panigrahi}%
 \email{pprasanta@iiserkol.ac.in}
\affiliation{National Institute of Technology Karnataka, Surathkal, Mangalore - 575025, Karnataka, India.}
\affiliation{Indian Institute of Science Education and Research Kolkata, Mohanpur - 741246, West Bengal, India.}

%


\date{\today}

\begin{abstract}
A geometric understanding of entanglement is proposed based on local measurements. Taking recourse to the general structure of density matrices in the framework of Euclidean geometry, we first illustrate our approach for bipartite Werner states. It is demonstrated that separable states satisfy certain geometric constraints that entangled states do not. A separability criterion for multiparty Werner states of arbitrary dimension is derived. This approach can be used to determine separability across any bipartition of a general density matrix and leads naturally to a computable measure of entanglement for multiparty pure states. It is known that all density matrices within a certain distance of the normalized identity are separable. This distance is determined for a general setting of $n$ qudits, each of dimension $d$. 
\end{abstract}

\pacs{Valid PACS appear here}
\maketitle


\section{\label{sec:level1}Introduction}

Quantifying and characterizing entanglement of a many body quantum system is an active area of research. Entanglement witnesses \cite{Horodecki, Terhal} have been identified for broad categorization of multiparty states. Apart from understanding the nature of correlation embodied in entanglement, the fact that it is an important resource in many quantum information and communication protocols makes its study significant \cite{nielsen, mosca}. Understanding of the geometry of the density matrices is very crucial for characterizing entanglement. \\
There have been various approaches to analyse the geometry of the quantum state space \cite{caves,monge,zyc} and also, entanglement measures based on geometry\cite{hey,ozawa,joag}. Recently, a measure of generalized concurrence has been proposed based on Lagrange wedge product \cite{Bhaskara}.  A Lie algebraic approach has naturally led to entanglement witness applying the positive partial transpose criterion on Casimir operator\cite{prasanta}.\\
\\
For convenience, the normalized matrix of order $N$ is at times denoted by $\mathbb{I}_N$ and the $N^2-1$ dimensional ball of which the convex set of density matrices of order $N$ is a subset, is denoted by $\mathbb{B}^{N^2-1}$. For example, in the two dimensional case, the Bloch ball is denoted by $\mathbb{B}^3$. An $n$ dimensional cross section of $\mathbb{B}^{N^2-1}$ is denoted by $B^n$.
\\
The Hilbert-Schmidt (Euclidean) distance between any two Hermitian matrices $\rho$ and $\sigma$ is given by:
\begin{equation*}
    D(\rho,\sigma) = \sqrt{\textrm{Tr}(\rho-\sigma)^2}
\end{equation*}
In this framework, the set of all density matrices of order $N=2$ is a closed ball of radius $\frac{1}{\sqrt{2}}$ centered at the normalized identity matrix $\frac{1}{2}\mathbb{I}$, called the Bloch ball. The boundary of this ball i.e., the Bloch sphere contains all $2\times 2$ density matrices of norm 1 i.e., the pure states. Antipodal points on the sphere represent orthogonal matrices. Each diameter of the Bloch sphere can thus be treated as a linearly independent or orthogonal basis. For $N>2$, the set of density matrices is no longer a closed ball. This set is however, always a convex set and is homeomorphic to a closed ball. All density matrices are Hermitian matrices. Therefore, a density matrix of order $N$ has $N$ real diagonal entries and $\frac{N(N-1)}{2}$ complex non-diagonal entries. The number of real parameters is thus $N^2$. The normalization constraint $\textrm{Tr}(\rho)=1$ reduces the number of real parameters by 1.  The set of density matrices of order $N$ is in general a compact set embedded in $N^2-1$ dimensional Euclidean space. This set always admits the regular $N-1$ simplex, which is an orthogonal basis just as a diameter of the Bloch sphere in the $N=2$ case.\\
\\
It is convenient to treat the set of density matrices of order $N$ as a subset of the closed $N^2-1$ dimensional ball of radius $\sqrt{\frac{N-1}{N}}$ centered at the normalized identity. The convex hull of a basis is represented by a regular $n$ simplex centered at $\frac{1}{N}\mathbb{I}$ and circumscribed by $\mathbb{B}^{N^2-1}$, where $N-1\le n\le N^2-1$. Each density matrix can thus be treated as a point in a simplex whose vertices are pure states. It follows that if a density matrix $\rho$ lies in a simplex whose vertices are separable states then $\rho$ is separable. Let $S$ be the set density matrix of a composite system. Then, an orthogonal basis spanned by pure separable states in $S$ can be `constructed' from orthogonal bases of the constituents of $S$. For example, the set of all two qubit density matrices is the set of density matrices of order 4, a subset of $\mathbb{B}^{15}$. We label the two qubits $A$ and $B$ for convenience. Consider two density matrices $\sigma_A$ and $\sigma^p_A$ orthogonal to each other i.e. they are antipodal points on the Bloch sphere of $A$ and similarly $\pi_B$ and $\pi^p_B$ on the Bloch sphere of $B$. The density matrices $\sigma_A\otimes\pi_B$, $\sigma_A\otimes\pi^p_B$, $\sigma^p_A\otimes\pi_B$ and $\sigma^p_A\otimes\pi^p_B$ are vertices of a regular tetrahedron circumscribed by $\mathbb{B}^{15}$ and centered at the normalized identity of order 4. Similarly, a regular 8-simplex circumscribed by $\mathbb{B}^{80}$ with pure separable two qutrit states as its vertices can be constructed from two equilateral triangles circumscribed by 7-dimensional spheres centered at normalized identities of order 3; a regular 15-simplex circumscribed by $\mathbb{B}^{15}$ with pure separable two qubit states as vertices  can be constructed from two regular tetrahedrons circumscribed by Bloch spheres, etc. It follows that all density matrices that lie in such constructed simplices are separable. Note that a regular simplex is the only convex set such that any given point can be uniquely expressed as a mixture of pure states\cite{beng}.\\
\\
The set of density matrices of order greater than 2 has a complex and interesting structure. We study this structure in the framework of Euclidean geometry and determine the necessary and sufficient conditions for a density matrix to be separable across any bipartition, by making use of local measurements and applying geometric constraints. We show that a two qudit Werner state is separable iff the fraction of the maximally entangled state $p$ satisfies $p\le \frac{1}{1+d}$ and an $n$ qudit Werner state is separable iff $p\le \frac{1}{1+d^{n-1}}$. Our results concur with those of \cite{rubin}. This approach leads to a computable measure of entanglement for pure multiparty states of arbitrary dimensions.

\begin{figure}
\centering
\begin{minipage}{.5\textwidth}
  \centering
  \includegraphics[width=.5\linewidth]{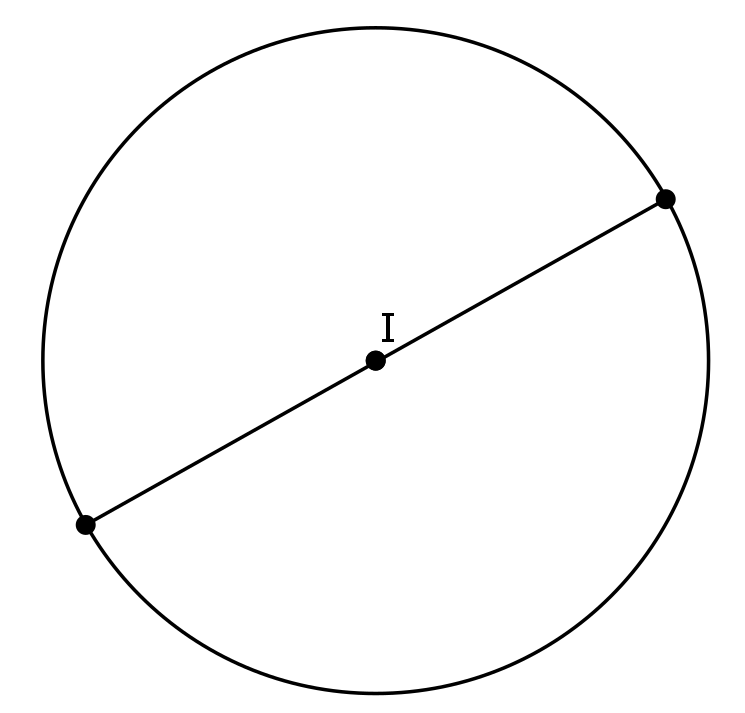}
  \caption{Orthogonal basis represented by a 1-simplex i.e., diameter of the Bloch ball in the $N=2$ case.}
  \label{fig:test1}
\end{minipage}%
\hfill
\begin{minipage}{.5\textwidth}
  \centering
  \includegraphics[width=.5\linewidth]{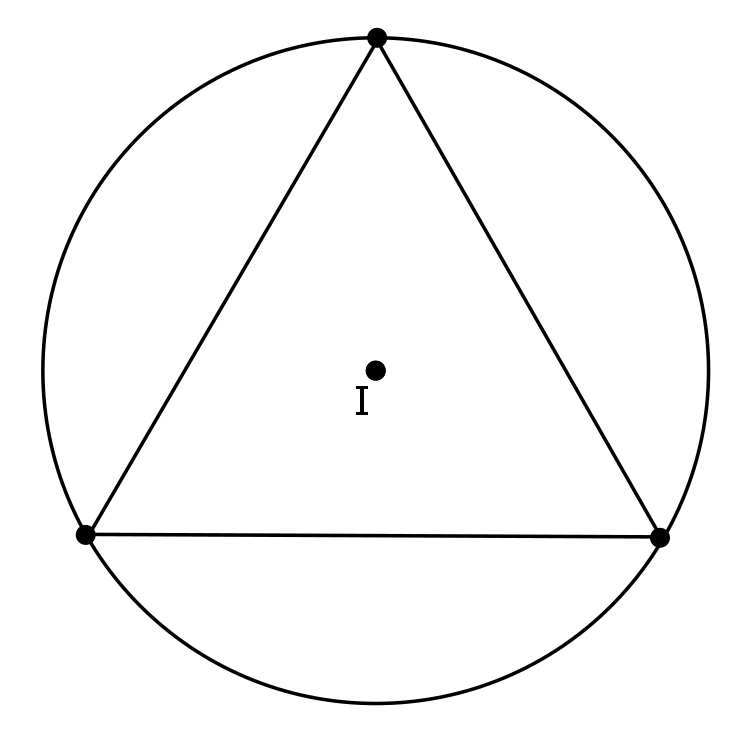}
  \caption{Orthogonal basis represented by a 2-simplex i.e., equilateral triangle in the $N=3$ case.}
  \label{fig:test2}
\end{minipage}
\begin{minipage}{.5\textwidth}
  \centering
  \includegraphics[width=.5\linewidth]{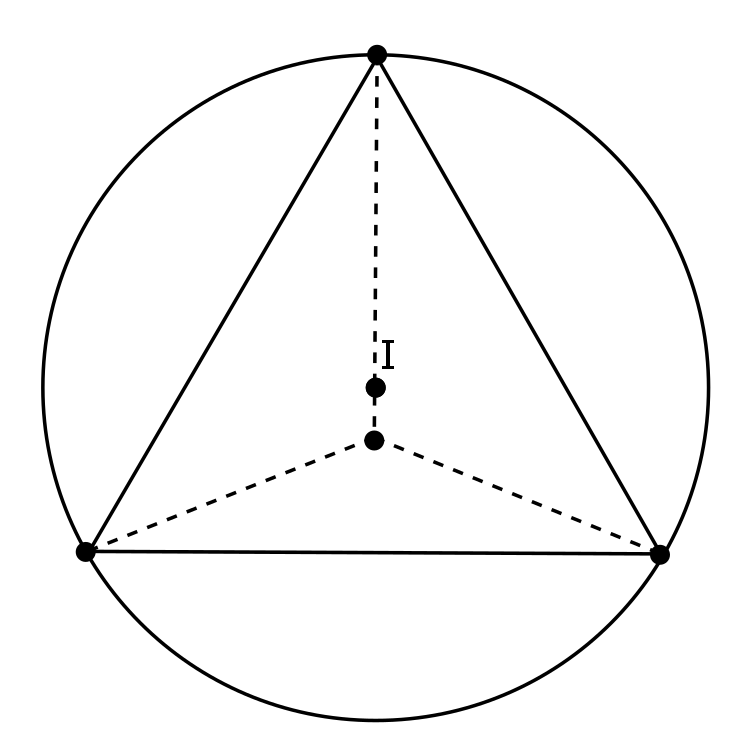}
  \caption{Orthogonal basis represented by a 3-simplex i.e., regular tetrahedron in the $N=4$ case.}
  \label{fig:test3}
\end{minipage}
\end{figure}

\section{Bipartite Werner states}

An $N\times N$ Werner state is a convex mixture of the $N\times N$ maximally entangled state and the  normalized identity of order $N^2$ i.e. $\rho_W= p\ket{\Psi}\bra{\Psi}+\frac{1-p}{N^2}\mathbb{I}$, where $\ket{\Psi}$ is the maximally entangled $N\times N$ state. Our aim is to determine the maximum fraction $p$ of the maximally entangled state, for which $\rho_W$ is separable. We first determine the maximum value of $p$ for the well known two qubit Werner state, by making use of local measurement and applying geometric constraint to the reduced density matrix of the unmeasured qubit. We then generalize the result to bipartite Werner states of arbitrary dimension.\\
\\
Consider a pure, separable two qubit state $\rho$ which can be expressed as the Kronecker product of the reduced density matrices of the constituent qubits i.e. $\rho=\rho_A\otimes\rho_B$. A projection of the qubit $A$ onto $M_A$ i.e. a local measurement, with the identity of order 2 acting on the qubit $B$ is expressed as $M\rho M^{\dagger}$, where $M=M_A\otimes\mathbb{I}_2$. The normalized post measurement state is given by $\sigma=\frac{\rho^{'}}{\textrm{Tr}(\rho^{'})}$, where $\rho^{'}=M\rho M^{\dagger}$. Clearly, the qubit $B$ is unaffected by the projection of $A$ onto $M_A$ i.e. $\textrm{Tr}_A(\sigma)=\sigma_B=\rho_B$ for all $M_A$, because there is zero entanglement between $A$ and $B$. This is not true in general, however. For example, consider a local measurement of one of the qubits of the separable state $\rho=\frac{1}{2}\ket{00}\bra{00}+\frac{1}{2}\ket{11}\bra{11}$, with the identity of order 2 acting on the other. If the qubit $A$ is projected onto a state $a\ket{0}+b\ket{1}$ with $|a|^2+|b|^2=1$, the post-measurement reduced density matrix of the qubit $B$ is given by $a^2\ket{0}\bra{0}+b^2\ket{1}\bra{1}$, which is a convex mixture of the orthogonal states $\ket{0}\bra{0}$ and $\ket{1}\bra{1}$. This implies that the post-measurement reduced density matrices of $B$ lie on a common diameter of the Bloch ball, whose endpoints are $\ket{0}\bra{0}$ and $\ket{1}\bra{1}$, for all $a,b\in\mathbb{C}$. These observations can be generalized as follows. Let $\rho$ be a separable density matrix of order $N$. Then, $\rho$ can be expressed as $\rho=\sum_i a_i\rho_i$ where $\sum_i a_i =1$ and the $\rho_i$ are pure separable states. The matrices $\{\rho_i\}$ form a \emph{basis}, not necessarily orthogonal. We consider a bipartition $A/B$ of $\rho$ such that $A$ is projected onto some state $M_A$ with the identity of appropriate dimension acting on $B$ i.e., $M\rho M^{\dagger}=M(\sum_i a_i\rho_i)M^{\dagger}=\sum_i a_iM\rho_iM^{\dagger}$, where $M=M_A\otimes\mathbb{I}_B$. This is rewritten as $\rho^{'}=\sum_i a_i\rho^{'}_i$, where $\rho^{'}=M\rho M^{\dagger}$ and $\rho^{'}_i=M\rho_i M^{\dagger}$. Taking partial trace over the measured party, $\rho^{'}_B =\sum_i a_i\rho^{'}_{iB}=\sum_i a_i\rho^{'}_{iB}\times\frac{\textrm{Tr}(\rho^{'}_{iB})}{\textrm{Tr}(\rho^{'}_{iB})}$. Since $\{\rho_i\}$ is a basis whose elements are pure separable states, $\frac{\rho^{'}_{iB}}{\textrm{Tr}(\rho^{'}_{iB})}=\rho_{iB}$ for all $i$. Therefore, $\rho^{'}_B = \sum_i a_i\textrm{Tr}(\rho^{'}_{iB})\rho_{iB}$ i.e $\sigma_B=\sum_i b_i\rho_{iB}$ where $\sigma_B=\frac{\rho^{'}_B}{\textrm{Tr}(\rho^{'}_B)}$. Since $\rho_B=\sum_i a_i\rho_{iB}$, both $\sigma_B$ and $\rho_B$ belong to the common convex hull spanned by the elements of the basis $\{\rho_{iB}\}$.\\
\\
Geometrically, a basis of $n$ density matrices of order $N$ is represented by a regular $n$ simplex centered at $\frac{1}{N}\mathbb{I}$ circumscribed by $\mathbb{B}^{N^2-1}$\cite{beng}. If a density matrix $\sigma$ is separable across a bipartition, a local projection of either party onto any pure state $\rho$ with the identity acting on the other reveals that the reduced density matrices of the unmeasured party belong to one or more common simplices spanned by pure states, over all $\rho$. For example, for a pure separable state, the simplex is just a point (0-simplex), for the state $\frac{1}{2}\ket{00}\bra{00}+\frac{1}{2}\ket{11}\bra{11}$, the simplex is a diameter of the Bloch ball (1-simplex) etc.\\
\\
Returning to the two qubit Werner state, a local measurement on any qubit can be expressed as  follows
\begin{equation*}
    M\rho_WM^{\dagger}=pM\ket{\Psi}\bra{\Psi}M^{\dagger}+\frac{1-p}{4}M\mathbb{I}M^{\dagger}
\end{equation*}
where $M=M_A\otimes\mathbb{I}_B$. Denoting $M\rho_WM^{\dagger}$ by $\rho^{'}_W$ and $M\ket{\Psi}\bra{\Psi}M^{\dagger}$ by $\sigma^{'}$, we have $\rho^{'}_W=p\sigma^{'}+\frac{1-p}{4}MM^{\dagger}$. Taking partial trace over the qubit $A$, 
\begin{equation*}
    \rho^{'}_{W_B}=p\sigma^{'}_B+\frac{1-p}{4}\mathbb{I}_2
\end{equation*}
where $\mathbb{I}_2$ is the identity matrix of order 2. $\textrm{Tr}(\rho^{'}_{W_B})=\frac{1}{2}$ independent of $M$. Thus, the post-measurement reduced density matrix of the qubit $B$ can be expressed as $\rho^m_{W_B}=p\sigma_B+\frac{1-p}{2}\mathbb{I}$, where $\rho^m_{W_B}=\frac{\rho^{'}_{W_B}}{\textrm{Tr}(\rho^{'}_{W_B}}$. The post-measurement reduced density matrices of $B$ lie on a sphere of radius $\frac{p}{\sqrt{2}}$ centered at $\frac{1}{2}\mathbb{I}$. In order for the Werner state to be separable, we require that all post measurement reduced density matrices of $B$ lie inside all regular 3-simplices i.e. tetrahedrons circumscribed by the Bloch ball and centered at $\frac{1}{2}\mathbb{I}$. This implies that $p$ should be less than or equal to the minimum ratio in which a tetrahedron divides a radius of the Bloch sphere or, equivalently, the quantity $\frac{p}{\sqrt{2}}$ should be less than or equal to the inradius of a regular tetrahedron inscribed in the Bloch ball. This yields the condition $p\le\frac{1}{3}$ for the two qubit Werner state to be separable.\\
\\
We similarly determine the maximum fraction of the maximally entangled two qutrit state for which the two qutrit Werner state is separable. The set of all density matrices of order 3 is a subset of the 8-dimensional ball of radius $\sqrt{\frac{2}{3}}$ centered at the normalized identity $\frac{1}{3}\mathbb{I}$. However, unlike the Bloch ball, not all points on the surface of the 8-dimensional ball are density matrices. The Werner state can be expressed as $\rho_W=p\ket{\Psi}\bra{\Psi}+\frac{1-p}{9}\mathbb{I}$. We again perform a local measurement on one of the qutrits and apply geometric constraint to the reduced density matrix of the other. The post-measurement reduced density matrix of the unmeasured qutrit can be expressed as $\rho^m_{W_B}=p\sigma_B+\frac{1-p}{3}\mathbb{I}$. Since not all points on the surface of $\mathbb{B}^8$ are density matrices, the $\rho^m_{W_B}$ lie along discrete radii, called \emph{proper radii} of $\mathbb{B}^8$ whose endpoints are $\frac{1}{3}\mathbb{I}$ and density matrices of order 3, and are at a distance of $p\sqrt{\frac{2}{3}}$ from the normalized identity. For separability, we require that all post-measurement reduced density matrices of $B$ lie inside all regular 8-simplices circumscribed by $\mathbb{B}^8$ and centered at $\frac{1}{3}\mathbb{I}$. \\
\\
The maximum fraction $p_{max}$ of the maximally entangled state for which $\rho_W$ is separable is equal to the minimum ratio in which an 8-simplex divides a proper radius of $\mathbb{B}^8$, as shown in figure 5. With reference to figure 6, the ratio $\frac{Ia}{Ib}$ is to be determined. By the property of similar triangles, $\frac{Ib}{Id}=\frac{Ia}{Ic}=2$. Recall that the ratio of the circumradius to the inradius of a regular $N$-simplex is equal to $N$. Thus, $\frac{Ib}{Ic}=8$ i.e. $Ib=8Ic$, and $Ia=2Ic$. This gives $\frac{Ia}{Ib}=\frac{1}{4}$. The maximum value of the fraction $p$ for which the two qutrit Werner state is separable is $\frac{1}{4}$, as is known\cite{cmc}. Extending this method, we find that the $4\times4$ Werner state is separable iff $p\le \frac{1}{5}$, the $5\times5$ Werner state is separable iff $p\le \frac{1}{6}$ and so on. In general,

\begin{lemma}
The maximum fraction of the $N\times N$ maximally entangled state in the $N\times N$ Werner state for which the Werner state is separable is $\frac{1}{N+1}$.
\end{lemma}

\begin{proof}
We again use the fact that if a local measurement is performed on an $N\times N$ separable Werner state, the reduced density matrices of the unmeasured qudit lie inside all $N^2-1$-simplices circumscribed by $\mathbb{B}^{N^2-1}$, centered at $\frac{1}{N}\mathbb{I}$. The maximum value of $p$ for which the Werner state is separable is equal to the minimum ratio in which the $N^2-1$-simplex divides a radius of $\mathbb{B}^{N^2-1}$ whose endpoints are $\frac{1}{N}\mathbb{I}$ and a density matrix. We employ the similar triangles approach as done for the two qutrit case. In general, $Ib=\sqrt{\frac{N-1}{N}}$ and $bd=\sqrt{\frac{N-2}{N-1}}$ which implies that $Id=\frac{1}{\sqrt{N(N-1)}}$. It is also known that $\frac{Ic}{Ib}=\frac{1}{N^2-1}$ i.e. $Ib=Ic(N^2-1)$. By the similar triangles property, $\frac{Id}{Ib}=\frac{Ic}{Ia}=\frac{1}{N-1}$ i.e. $Ia=Ic(N-1)$. Therefore, we have $\frac{Ia}{Ib}=\frac{N-1}{N^2-1}=\frac{1}{N+1}$.\\
\end{proof}

Since an $N^2-1$ simplex has $N^2$ vertices, a separable mixed state $\rho$ of order $N$ can have upto $N^2$ pure product states in its decomposition. More specifically, a separable full rank density matrix of order $N$ has atleast $N$ and at most $N^2$ pure product state in its decomposition\cite{uhl}. 

\begin{figure}
\centering
\begin{minipage}{.5\textwidth}
  \centering
  \includegraphics[width=.5\linewidth]{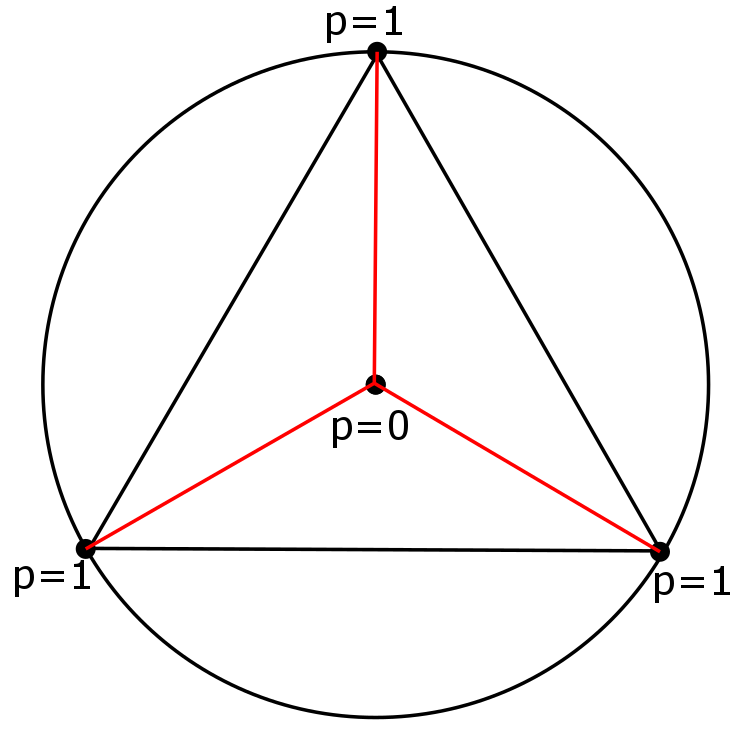}
  \caption{A 2D cross section of $\mathbb{B}^8$. The post local measurement reduced density matrix of  $B$ moves along the red radii as $p$ varies from 0 to 1.}
  \label{fig:test4}
\end{minipage}%
\hfill
\begin{minipage}{.5\textwidth}
  \centering
  \includegraphics[width=.5\linewidth]{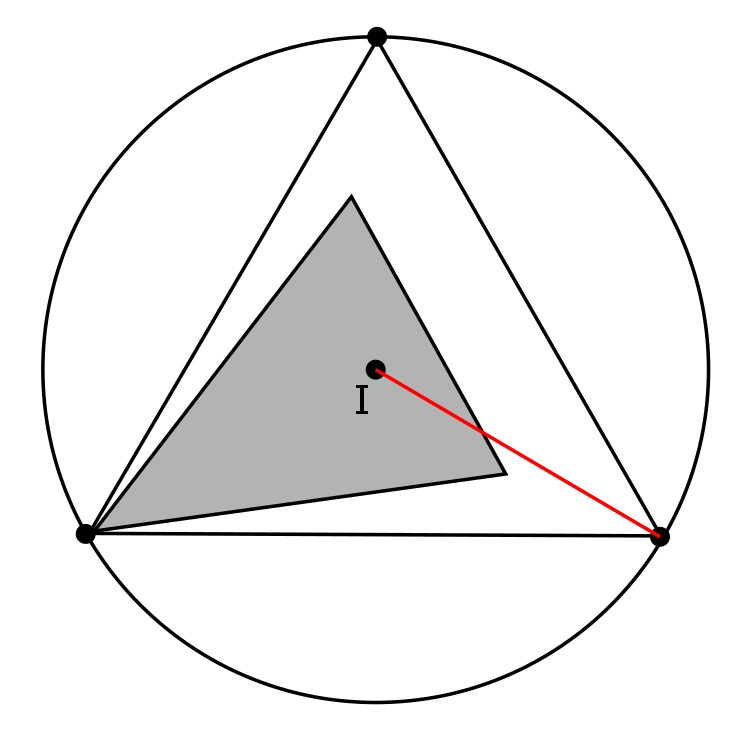}
  \caption{The grey region is a cross section of an 8-simplex and the ratio in which the boundary of the gray region divides the radius is the maximum value of $p$ for which $\rho_W$ is separable.}
  \label{fig:test5}
\end{minipage}
\hfill
\begin{minipage}{.5\textwidth}
  \centering
  \includegraphics[width=.5\linewidth]{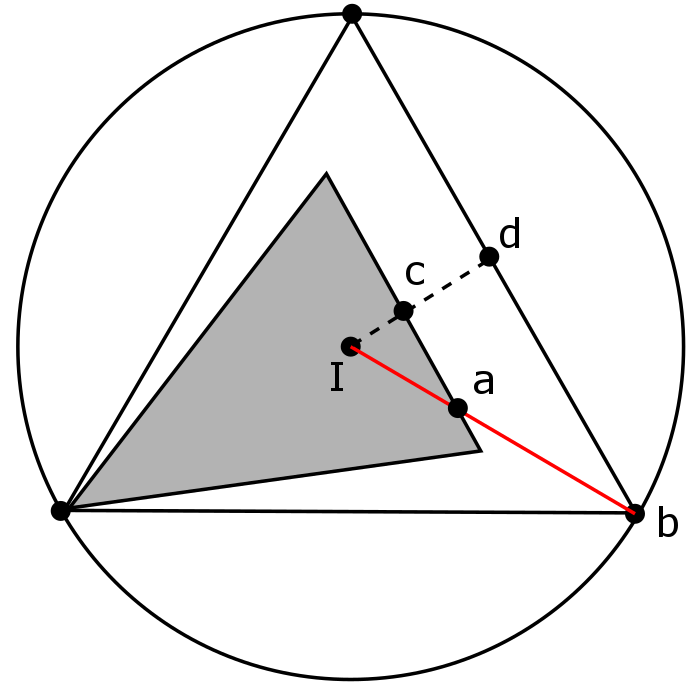}
  \caption{$\frac{Ia}{Ib}$ is the minimum ratio in which an 8-simplex divides a proper radius of $\mathbb{B}^8$.}
  \label{fig:test6}
\end{minipage}
\end{figure}



\section{Multiparty Werner states}
 
 Having derived the necessary and sufficient condition for a bipartite $N\times N$ Werner state to be separable, the natural question would be if the separability criterion for multiparty Werner states can be derived similarly, using geometric constraints.\\
 We begin by determining the condition for separability of the three qubit Werner state and then generalize the result to an $n$-qudit setting.\\
 \\
 The three qubit Werner state is a mixture of the $GHZ$ state and the normalized identity matrix of order 8 i.e. $\rho_W = p\ket{GHZ}\bra{GHZ}+\frac{1-p}{8}\mathbb{I}$. Upon performing a local measurement on the $GHZ$ state on any one of the qubits, it is clear that the probability corresponding to the measurement and distance between the pre and post local measurement reduced density matrices of the joint system of the unmeasured qubits are $\frac{1}{2}$ and $\frac{1}{\sqrt{2}}$ respectively, independent of the measurement operator. The $GHZ$ state thus has maximal entanglement across each bipartition, equal in magnitude to that of the Bell state. The $GHZ$ state is thus the maximally entangled three qubit state. Interestingly, although there exists maximal entanglement between one qubit and the joint system of the other two, there exists \emph{zero} entanglement between any two qubits of the $GHZ$ state.\\
 \\
 Our aim is again to determine the maximum value of the fraction $p$ for which $\rho_W$ is separable. For convenience, we label the qubits $A$, $B$ and $C$. Let $M_A$ be the pure state on to which the qubit $A$ is projected, with the identity acting on the joint system $BC$.
 \begin{equation*}
     M\rho_W M^{\dagger}= pM\ket{GHZ}\bra{GHZ}M^{\dagger}+ \frac{1-p}{8}M\mathbb{I}M^{\dagger}
 \end{equation*}
 where $M=M_A\otimes\mathbb{I}_{BC}$. Denoting $M\rho_W M^{\dagger}$ by $\rho^{'}_W$ and $M\ket{GHZ}\bra{GHZ}M^{\dagger}$ by $\sigma^{'}$, we have $\rho^{'}_W= p\sigma^{'}+ \frac{1-p}{8}MM^{\dagger}$. Taking partial trace over qubit $A$,
 \begin{equation*}
     \rho^{'}_{W_{BC}} = p\sigma^{'}_{BC} + \frac{1-p}{8}\mathbb{I}_4
 \end{equation*}
 where $\mathbb{I}_4$ is the identity matrix of order 4. $\textrm{Trace}(\rho^{'}_{W_{BC}})=\frac{1}{2}$, independent of $M$. The normalized post-measurement reduced density matrix of $BC$ can be expressed as $\rho^m_{W_{BC}}=p\sigma_{BC}+\frac{1-p}{4}\mathbb{I}_4$, where $\sigma_{BC}= 2\sigma^{'}_{BC}$. Now, the distance between $\rho^m_{W_{BC}}$ and $\frac{1}{4}\mathbb{I}$ is $p\sqrt{\frac{3}{4}}$, also independent of $M$, implying that all post-measurement reduced density matrices of the joint system $BC$ are equidistant from the center. Thus, in order to determine the maximum value of $p$ for which $\rho_W$ is separable, we must determine the minimum ratio in which a 15-simplex circumscribed by $\mathbb{B}^{15}$ divides a proper radius of $\mathbb{B}^{15}$. From the results of the previous section, it can be concluded that $p\le\frac{1}{5}$ for $\rho_W$ to be separable i.e., the maximum fraction of the $GHZ$ state in the three qubit Werner state for which the Werner state is separable, is $\frac{1}{5}$.\\
 \\
 Consider an $n$-qubit Werner state $\rho_W$ which is a mixture of the $GHZ$ state $\frac{1}{\sqrt{2}}\ket{000...0}+\frac{1}{\sqrt{2}}\ket{111...1}$ and the normalized identity of order $2^n$. A bipartition of the $GHZ$ state consists of two parties $A$ and $B$ of $m$ and $k$ qubits such that $n=m+k$. There is non-zero and maximal entanglement across the bipartition iff $m=1, k=n-1$ or $m=n-1, k=1$. Therefore, to determine the separability criterion for $\rho_W$, it is sufficient to perform a local measurement on $\rho_W$  on any one qubit and apply geometric constraints to the post measurement reduced density matrix of the joint system of $n-1$ qubits.\\
 Just as in the three qubit case, the post measurement reduced density matrices of the unmeasured system of $n-1$ qubits are equidistant, at a distance of $p\sqrt{\frac{2^{n-1}-1}{2^{n-1}}}$ from the center $\frac{1}{2^n-1}\mathbb{I}$, where $p$ is the fraction of the $GHZ$ state in the Werner state. Tha maximum value of $p$ for which $\rho_W$ is separable is equal to the minimum ratio in which the $2^{2(n-1)}-1$ simplex circumscribed by $\mathbb{B}^{2^{2(n-1)}-1}$ divides a proper radius. Thus, substituting $N=2^{n-1}$ in Lemma 3, we make the following observation
 \begin{observation}
 The maximium fraction of the $n$-qubit $GHZ$ state in the $n$-qubit Werner state, for which the Werner state is separable, is $\frac{1}{2^{n-1}+1}$. 
 \end{observation}
 
 We now generalize these results to a system of $n$ particles, each of dimension $d$ i.e. an $n$-qudit system. An $n$ qudit Werner state is a mixture of the maximally entangled $n$-qudit state $\sum_{i=1}^n \frac{1}{\sqrt{d}}\ket{iii...i}$ and the normalized identity of order $d^n$ i.e. $\rho_W = p\ket{\Psi}\bra{\Psi}+\frac{1-p}{d^n}\mathbb{I}$. We denote by $A$ the qudit to be measured and by $B$ the joint system of the $n-1$ qudits on which the identity of order $d^{n-1}$ acts. Let $M=M_A\otimes\mathbb{I}_B$ be the measurement operator. Then
 \begin{equation*}
     M\rho_W M^{\dagger} = pM\ket{\Psi}\bra{\Psi}M^{\dagger}+ \frac{1-p}{d^n}M\mathbb{I}M^{\dagger}
 \end{equation*}
Let $M\rho M^{\dagger}=\rho^{'}_W$ and $M\ket{\Psi}\bra{\Psi}M^{\dagger}=\sigma^{'}$. Taking partial trace over the qudit $A$, we have 
\begin{equation*}
    \rho^{'}_{W_B} = p\sigma^{'}_B +\frac{1-p}{d^n}\mathbb{I}_B
\end{equation*}
 where $\mathbb{I}_B$ is the normalized identity matrix of order $d^{n-1}$. The probability corresponding the measurement is $\frac{1}{d}$, independent of $M$. The normalized post-measurement reduced density matrix of the unmeasured system can be expressed as $\rho^m_{W_B}=p\sigma_B+\frac{1-p}{d^{n-1}}\mathbb{I}$, where $\sigma_B=d\sigma^{'}_B$. The post-measurement reduced density matrices of $B$ are at a distance of $p\sqrt{\frac{d^{n-1}-1}{d^{n-1}}}$, independent of $M$. Therefore, the maximum value of $p$ for which $\rho_W$ is separable is equal to the minimum ratio in which a $d^{2(n-1)}-1$ simplex circumscribed by $\mathbb{B}^{d^{2(n-1)}-1}$ divides a proper radius. We formally state the necessary and sufficient condition for an $n$-qudit Werner state to be separable
 \begin{lemma}
 The maximum fraction of the maximally entangled $n$-qudit state in the $n$-qudit Werner state for which the Werner state is separable is $\frac{1}{d^{n-1}+1}$
 \end{lemma}
 
We have thus derived geometrically the necessary and sufficient criterion for the separability of general Werner states, as promised. Similarly, the separability across any bipartition of a density matrix can be tested. 
 
 \section{Entanglement measure for pure states}

The ideas developed in the preceding sections can be extended to compute the entanglement across any bipartition of a pure, multiparty state of arbitrary dimension. Let $\sigma=\ket{\Phi}\bra{\Phi}$ be a density matrix of order $N$. A density matrix lying on the radius of $\mathbb{B}^{N^2-1}$ whose endpoints are $\sigma$ and $\frac{1}{N}\mathbb{I}$ can be expressed as $\rho=p\sigma+\frac{1-p}{N}\mathbb{I}$, where $0\le p\le 1$. All such density matrices have rank $N$ for $p<1$. Suppose the entanglement across some bipartition $A/B$ of $\sigma$ is to be computed. We first determine the maximum value of $p$ for which $\rho$ is separable, by performing a local measurement on $A$. Note that a separable density matrix of order $N$ may lie inside a unique $N$ simplex spanned by separable states but it lies inside infinitely many $N^2-1$ simplices since an $N^2-1$ simplex represents a non-orthogonal basis. Let $M=M_A\otimes\mathbb{I}_B$ be the measurement operator acting on $\rho$. Then, $M\rho M^{\dagger}=pM\sigma M^{\dagger}+\frac{1-p}{N}M\mathbb{I}M^{\dagger}$. Following the usual notation, $\rho^{'}_B=p\sigma^{'}_B+\frac{1-p}{N}\mathbb{I}_B$. Let the subsystem $A$ have dimension $d$. Then $\rho^{'}_B=p\sigma^{'}_B\times\frac{\textrm{Tr}(\sigma^{'}_B}{\textrm{Tr}(\sigma^{'}_B}+\frac{1-p}{N}\mathbb{I}\times\frac{\frac{N}{d}}{\frac{N}{d}}=p\textrm{Tr}(\sigma^{'}_B)\pi_B+\frac{1-p}{d}\mathbb{I}_n$, where $\pi_B=\frac{\sigma^{'}_B}{\textrm{Tr}(\sigma^{'}_B}$ and $\mathbb{I}_n$ is the normalized identity matrix of order $\frac{N}{d}$. $\textrm{Tr}(\rho^{'}_B)=pq+\frac{1-p}{d}$ where $q=\textrm{Tr}(\sigma^{'}_B)$. The post-measurement reduced density matrix of $B$ can thus be expressed as $\rho^m_B=P\pi_B+\frac{d(1-P)}{N}\mathbb{I}$. We require that $\rho^m_B$ lie inside all (infinitely many) $(\frac{N}{d})^2-1$ simplices that contain $\rho_B$, for $\rho$ to be separable across the bipartition $A/B$. This implies the following condition
\begin{equation*}
    P=\frac{pq}{pq+\frac{1-p}{d}}\le\frac{1}{1+\frac{N}{d}}
\end{equation*}
which upon simplification yields a condition on $p$, as desired.
\begin{equation*}
    p\le\frac{1}{1+Nq}
\end{equation*}
The maximum value of the fraction $p$ for which $\rho$ is separable is therefore $p_{max}=\frac{1}{1+Nq_{min}}$. If $\sigma$ is a separable state then $q_{min}=0$ which corresponds to the projection of $\sigma_A$ onto a state orthogonal to it, implying that $p\le 1$ i.e. $\rho$ is separable for all $p$. If $\sigma$ is a maximally entangled $n$-qudit state then $N=d^n$ and $q=q_{min}=\frac{1}{d}$ independent of the state onto which $\sigma_A$ is projected, which gives $p_{max}=\frac{1}{1+d^{n-1}}$, as earlier determined.\\
\\
Given any pure state $\sigma$, the maximum value of $p$ for which $p\sigma+\frac{1-p}{N}\mathbb{I}$ is separable is an indicator of the entanglement across some bipartition of $\sigma$. In fact, the quantity
\begin{equation*}
    e=1-p_{max}
\end{equation*}
is a computable measure of entanglement across any bipartition of a multiparty pure state. $e=0$ for pure, separable states since $p_{max}=1$. As an example, consider the $\ket{W}$ state which, in the computational basis can be expressed as $\ket{W}=\frac{1}{\sqrt{3}}\ket{001}+\frac{1}{\sqrt{3}}\ket{010}+\frac{1}{\sqrt{3}}\ket{100}$. Like the $GHZ$ state, the $W$ state is invariant under permutations of qubits, hence there is equal entanglement across all bipartitions of the $W$ state. The maximum value of $p$ for which the state $p\ket{W}\bra{W}+\frac{1-p}{8}\mathbb{I}$ is separable, is to be determined. The minimum probability of projection of any qubit onto a pure state, over all projections is found to be $q_{min}=\frac{1}{3}$. This gives $p_{max}=\frac{1}{1+Nq_{min}}=\frac{3}{11}$. Thus, the entanglement across every bipartition of the $W$ state is equal to $e=1-p_{max}=\frac{8}{11}$, which is slightly less than the corresponding value of $\frac{4}{5}$ for the $GHZ$ state. However, there is \emph{zero} entanglement present between every pair of qubits of the $GHZ$ state while there is non-zero entanglement present between every pair of qubits of the $W$ state. The correlations between constituents of a multiparty state can be represented by a \emph{graph} or even a \emph{hypergraph}\cite{hyper}. These representations provide a deeper understanding of the `structure' of entanglement in a multiparty state.\\ 
\\
From the geometry of the set of density matrices, it is clear that if a density matrix is close enough to the normalized identity, it is guaranteed to be separable\cite{vss}. What is the maximum distance a density matrix of order $N$ can be from $\frac{1}{N}\mathbb{I}$ so that it is guaranteed to be separable? The distance of any density matrix $\rho$ of order $N$ from $\frac{1}{N}\mathbb{I}$ i.e. the quantity $||\rho-\frac{1}{N}\mathbb{I}||$ can be expressed as $p\sqrt{\frac{N-1}{N}}$ where $0\le p\le 1$. The maximum value of $p$ such that any density matrix $\rho$ that satisfies $||\rho-\frac{1}{N}\mathbb{I}||\le p\sqrt{\frac{N-1}{N}}$ is necessarily separable, is to be determined. For a pure state $\sigma$, the maximum value of $p$ for which $p\sigma+\frac{1-p}{N}\mathbb{I}$ is separable across a bipartition is given by $p_{max}=1-e$ where $e$ is the entanglement across the bipartition. The minimum value of $p_{max}$ occurs when $e$ is maximized i.e. when $\sigma$ is a maximally entangled state. In an $n$-qudit setting, the minimum value of $p_{max}$ is thus $\frac{1}{1+d^{n-1}}$. This implies the following important observation

\begin{observation}
All $n$-qudit density matrices that are within $\frac{1}{1+d^{n-1}}\sqrt{\frac{d^n-1}{d^n}}$ distance of the normalized identity $\frac{1}{d^n}\mathbb{I}$ are separable.
\end{observation}

The set of all density matrices of order $d^n$ whose distance from $\frac{1}{d^n}\mathbb{I}$ is less than or equal to $\frac{1}{1+d^{n-1}}\sqrt{\frac{d^n-1}{d^n}}$ are \emph{absolutely separable} i.e. they cannot become entangled even by the action of global unitary operators\cite{nirman}. \\

\section{Concluding remarks}

To summarize, we determined the necessary and sufficient criterion for the separability of an $n$ qudit Werner state by performing a local measurement on any single qudit and requiring that the reduced density matrices of the unmeasured party lie inside all regular $d^{2(n-1)}-1$ simplices  circumscribed by $\mathbb{B}^{d^{2(n-1)}-1}$ over all measurements. The methods employed can be used to determine separability across any bipartition of a density matrix. This approach leads naturally to a computable measure of entanglement for pure multiparty states. Given any pure density matrix $\sigma$ of order $N$, the maximum value of the fraction $p$ for which $p\sigma +\frac{1-p}{N}\mathbb{I}$ is separable across some bipartition is first determined. Then, the entanglement $e(\sigma)$ across the bipartition is given by $e(\sigma)=1-p_{max}(\sigma)$ and is equivalent to detrmining the distance between $\sigma$ and the separable state closest to $\sigma$ along the radius containing $\sigma$. The maximum distance that a density matrix can be from the normalized identity for it to be guaranteed separable is also determined.\\
Throughout this letter it has been assumed that the set of density matrices of order $N$ admits the $N^2-1$ simplex for all $N$. These simplices are called SIC-POVMs for short and their existence for all $N$ has been conjectured but is known to be true only for $N\le 16$. The fact that the conclusions drawn from our method agree with those of\cite{rubin} strongly indicates that SIC-POVMs always exist. \\
We hope that our work can be extended to arrive at a measure of entanglement for mixed multiparty states and that it is of relevance to various related problems, such as entanglement in continuous variable systems\cite{duan,simon} and motivates further research in the field.\\
\\
AP would like to thank Clive Emary (Newcastle University), Subhayan Roy Moulick (IISER Kolkata) and Ashutosh K. Goswami (IISER Kolkata) for helpful comments and discussions.



\end{document}